\documentclass[a4paper,12pt]{article}

\usepackage{amssymb,amsmath,amsfonts,amsthm,geometry,graphicx,caption,xcolor,setspace,array,mathtools}
\usepackage[authoryear,round]{natbib}
\usepackage[hidelinks]{hyperref}
\allowdisplaybreaks[2]
\usepackage[inline]{enumitem}
\usepackage{tikz}
\usepackage{float}

\theoremstyle{plain}
\newtheorem{thm}{Theorem}
\newtheorem{lem}{Lemma}

\newtheorem{prop}{Proposition}

\theoremstyle{definition}

\newtheorem{emp}{Example}
\newtheorem*{prf}{Proof}
\newtheorem*{prfofthm1}{Proof of Theorem~\ref{thm:lattice}}

\DeclareMathOperator{\CR}{\textit{CR}}
\DeclareMathOperator{\TTC}{\textit{TTC}}
\DeclareMathOperator{\NE}{\mathbf{NE}}

\newcolumntype{L}[1]{>{\raggedright\let\newline\\arraybackslash\hspace{0pt}}m{#1}}
\newcolumntype{C}[1]{>{\centering\let\newline\\arraybackslash\hspace{0pt}}m{#1}}
\newcolumntype{R}[1]{>{\raggedleft\let\newline\\arraybackslash\hspace{0pt}}m{#1}}

\begin{document}

\title{The lattice structure of the crawler\thanks{This work was supported by JST ERATO (JPMJER2301).
}}

\author{Hinata Kurashita\thanks{Department of Industrial Engineering and Economics, Graduate School of Engineering, Institute of Science Tokyo, 2-12-1 Ookayama, Meguro-ku, Tokyo 152-8552, Japan. e-mail: kurashita.h.aa@m.titech.ac.jp.}}

\date{\today}

\maketitle

\begin{abstract}
    We study object reallocation problems with single-peaked preferences. On the unrestricted domain of strict preferences, top trading cycles (TTC) is the unique rule that satisfies \emph{efficiency}, \emph{individual rationality}, and \emph{strategy-proofness}. 
    On the single-peaked preference domain, the crawler introduced by \citet{bade2019matching} satisfies the same properties.
    We show that the Nash equilibrium allocations of the crawler’s direct
    revelation game form a complete lattice under weak Pareto dominance. The truthful allocation is the greatest element of this lattice, while the initial endowment allocation is its least element.
\end{abstract}


\doublespacing
\setcounter{page}{1}
\pagenumbering{arabic}

\section{Introduction}\label{sec:introduction}
We consider the reallocation of indivisible objects without monetary transfers.
Each agent initially owns one object and receives one object after reallocation.
Examples include users exchanging reserved time slots for shared facilities or equipment, workers exchanging assigned shifts of equal length, and residents exchanging assigned parking spaces along a street.
Reservations and shifts can be ordered by starting time, while parking spaces can be ordered by location.
When each agent has an ideal time or location and strictly prefers alternatives closer to it on the same side, preferences are \emph{single-peaked} with respect to this order.
The initial assignments may not reflect agents' preferences.
In such cases, exchanges may improve their assignments.

For given initial endowments, a \emph{rule} selects an allocation for each preference profile.
An allocation assigns one object to each agent and each object to exactly one agent.
We consider three standard properties of rules.
\emph{Efficiency} requires that no other allocation make some agent strictly better off without making another agent worse off.
\emph{Individual rationality} requires that each agent weakly prefer their assignment to their initial endowment.
\emph{Strategy-proofness} requires that each agent weakly prefer the assignment obtained by reporting truthfully to that obtained by any misreport, regardless of the other agents' reports.

\emph{Top trading cycles (TTC)}, introduced by \citet{shapley1974cores}, satisfies these three properties.
At each step, each remaining agent points to the owner of their most preferred remaining object.
In each cycle, each agent receives the object owned by the agent to whom they point.
These agents leave with their assigned objects, and the procedure repeats until no agent remains.
\citet{ma1994strategy} shows that, on the unrestricted domain of strict preferences, TTC is the unique rule satisfying \emph{efficiency}, \emph{individual rationality}, and \emph{strategy-proofness}.

This uniqueness result does not hold on the single-peaked preference domain.
The \emph{crawler}, introduced by \citet{bade2019matching}, also satisfies these three properties on this domain.\footnote{For further studies of the crawler, see \citet{tamura2022object} and \citet{tamura2022crawler}.}
The crawler considers agents in the order of their current holdings.
The first agent whose most preferred remaining object is weakly to the left of their holding receives that object.
Consider the other agents whose holdings are weakly to the right of the assigned object and strictly to the left of the selected agent's holding.
Each of these agents receives the next remaining object to their right as their new holding.
The selected agent and object leave, and the procedure repeats.

Although the crawler and TTC share these three properties, they may select different allocations.
What, then, distinguishes the two rules on the single-peaked preference domain?
\citet{bade2019matching} provides one answer: the crawler admits an \emph{obviously strategy-proof implementation}, whereas TTC does not in general.\footnote{On the unrestricted domain of strict preferences, TTC does not admit an obviously strategy-proof implementation when
there are at least three agents; see \citet{li2017obviously}.
On the single-peaked preference domain, however, it admits such an implementation when there are at most three agents; see \citet{bade2019matching}.}
\emph{Obvious strategy-proofness}, introduced by \citet{li2017obviously}, allows agents to recognize why they should follow the prescribed strategy without detailed reasoning about the behavior of others.\footnote{In an extensive-form mechanism, a strategy is \emph{obviously dominant} if, for each alternative strategy and each information set where the two strategies first prescribe different actions, the worst possible outcome under the prescribed strategy is weakly preferred to the best possible outcome under the alternative, conditional on that information set.
An obviously strategy-proof implementation assigns such a strategy to each agent for each preference and implements the rule when agents follow these strategies.}
Thus, this distinction concerns an incentive requirement stronger than \emph{strategy-proofness}.

In contrast, we study the Nash equilibrium allocations of the crawler's direct revelation games, in which agents simultaneously report their preferences.
\emph{Strategy-proofness} makes the truthful report profile a dominant-strategy equilibrium.
Nash equilibrium is a weaker concept: no agent can gain by changing their report alone.

We find a surprising property of the equilibrium allocations.
For each true preference profile, the crawler's Nash equilibrium allocations form a complete lattice under weak Pareto dominance (Theorem \ref{thm:lattice}).
Here, an allocation weakly Pareto dominates another if each agent weakly prefers their assigned object in the former to that in the latter.
Its greatest element is the truthful allocation, and its least element is the initial endowment allocation.
Thus, the truthful allocation weakly Pareto dominates each equilibrium allocation, and each equilibrium allocation weakly Pareto dominates the initial endowment allocation.
Moreover, for any two equilibrium allocations, assigning each agent the preferred object among their two assigned objects yields another equilibrium allocation, their join.
Choosing the less preferred object likewise yields their meet.
By contrast, these properties need not hold for TTC (Example \ref{emp:ttc-nonlattice}).
The structure of Nash equilibrium allocations therefore provides a further distinction between the two rules.

Interestingly, lattice structure and dominance of the truthful allocation have also been shown in a different model.
In the model studied by \citet{sprumont1991division}, a fixed amount of an infinitely divisible object is allocated among agents with single-peaked preferences over their individual quantities.
For this model, \citet{bochet2020dominance} show that, under each rule satisfying \emph{efficiency}, \emph{strategy-proofness}, and \emph{replacement monotonicity}, the Nash equilibrium allocations form a complete lattice whose greatest element is the truthful allocation.\footnote{\emph{Replacement monotonicity} requires that, when one agent changes their report and receives a weakly larger amount, each other agent receives a weakly smaller amount.
\citet{barbera1997strategy} show that these three properties characterize the class sequential allotment rules.}
For the uniform rule,\footnote{The uniform rule belongs to the class of sequential allotment rules.} the least element is the equal division allocation.\footnote{See also the online appendix to
\citet{bochet2020dominance}.}
The two models thus have a similar structure of equilibrium allocations, but this paper studies the reallocation of indivisible objects, whereas \citet{bochet2020dominance} study the allocation of an infinitely divisible object.

The rest of the paper is organized as follows.
Section~\ref{sec:preliminaries} introduces the model and the relevant definitions.
Section~\ref{sec:result} states the main result.
Section~\ref{sec:conclusion} provides concluding remarks.
The proofs are collected in the appendix.

\section{Preliminaries}\label{sec:preliminaries}
\subsection{Model}\label{subsec: model}
Let $N=\{1,\ldots,n\}$ be a finite set of agents and $O$ a finite set of objects, with $|O|=n$. 
Each agent $i\in N$ is initially endowed with an object $\omega_i\in O$, and no two agents own the same object. 
Let $\omega=(\omega_i)_{i\in N}$ denote the initial  endowment profile.
The set $O$ is equipped with a strict linear order $<$. 
Each agent $i\in N$ has a complete, transitive, and antisymmetric preference relation $R_i$ on $O$. 
Its strict part is denoted by $P_i$.\footnote{For $o,o'\in O$,
$o\mathbin{P_i}o'$ if and only if $o\mathbin{R_i}o'$ and $o\ne o'$.}
A preference $R_i$ is \textbf{single-peaked} if there exists $p_i\in O$ such that, for each $o,o'\in O$, $o'\mathbin{P_i}o$ if $o<o'\le p_i$, and $o\mathbin{P_i}o'$ if $p_i\le o<o'$.
Let $\mathcal{R}$ denote the set of all single-peaked preferences. 
A preference profile is $R=(R_i)_{i\in N}\in\mathcal{R}^n$.
For each $i\in N$, let $R_{-i}=(R_j)_{j\in N\setminus\{i\}}$, and for each $N'\subset  N$, let $R_{-N'}=(R_j)_{j\in N\setminus N'}$.

An \textbf{allocation} is a vector $x=(x_i)_{i\in N}\in O^n$ with $x_i\ne x_j$ for distinct $i,j\in N$.
Let $\mathcal{X}$ denote the set of all allocations.
For $i\in N$, $x_i$ denotes the object assigned to $i$.
For $o\in O$, $x_o$ denotes the agent assigned to $o$.

Throughout, $N$, $O$, $\omega$, and $<$ are fixed, with agents indexed so that $\omega_1<\omega_2<\cdots<\omega_n$.
A \textbf{rule} is a mapping $f\colon\mathcal{R}^n\to\mathcal{X}$.
For each $R\in\mathcal{R}^n$, $f(R)$ 
is the allocation selected by $f$, and $f_i(R)$ denotes the object assigned to agent $i$.

\subsection{Lattice structure}\label{subsec:lattice}
Fix $R\in\mathcal{R}^n$.
For $x,y\in\mathcal{X}$, $x$ \textbf{weakly Pareto dominates}
$y$, denoted by $x\succeq_R y$, if $x_iR_i y_i$ for each
$i\in N$.
The relation $\succeq_R$ is a partial order on
$\mathcal{X}$.\footnote{A partial order is a reflexive,
transitive, and antisymmetric binary relation.}

Let $\emptyset\ne\mathcal{Y}\subseteq\mathcal{X}$.
For $\mathcal{Y}'\subseteq\mathcal{Y}$, $z\in\mathcal{Y}$
is an upper bound of $\mathcal{Y}'$ if $z\succeq_R x$
for each $x\in\mathcal{Y}'$, and a lower bound if
$x\succeq_R z$ for each $x\in\mathcal{Y}'$.
The \textbf{join} $\bigvee\mathcal{Y}'$ is an upper bound
weakly Pareto dominated by each upper bound in $\mathcal{Y}$.
The \textbf{meet} $\bigwedge\mathcal{Y}'$ is a lower bound
that weakly Pareto dominates each lower bound in $\mathcal{Y}$.
These bounds are unique whenever they exist.

The partially ordered set $(\mathcal{Y},\succeq_R)$ is a
\textbf{lattice} if each pair has a join and a meet.
It is \textbf{complete} if each subset, including the empty set, has a join and a meet.
Its \textbf{greatest element} is an allocation in $\mathcal{Y}$
that weakly Pareto dominates each allocation in $\mathcal{Y}$.
Its \textbf{least element} is an allocation in $\mathcal{Y}$
weakly Pareto dominated by each allocation in $\mathcal{Y}$.

For $x,y\in O^n$, define $x\vee y,x\wedge y\in O^n$ as follows. For each $i\in N$,
\begin{equation*}
[x\vee y]_i=
\begin{cases}
x_i&\text{if $x_i\mathbin{R_i} y_i$},\\
y_i&\text{otherwise},
\end{cases}
\end{equation*}
and
\begin{equation*}
[x\wedge y]_i=
\begin{cases}
x_i&\text{if $y_i\mathbin{R_i} x_i$,}\\
y_i&\text{otherwise}.
\end{cases}
\end{equation*}
Note that these vectors need not belong to $\mathcal{Y}$,
even when $x,y\in\mathcal{Y}$.
If both vectors belong to $\mathcal{Y}$ for each $x,y\in\mathcal{Y}$, then
$\bigvee\{x,y\}=x\vee y$ and $\bigwedge\{x,y\}=x\wedge y$:
for any common upper bound $v$ and common lower bound $w$ in $\mathcal{Y}$,
$v_i\mathbin{R_i}[x\vee y]_i$ and $[x\wedge y]_i\mathbin{R_i}w_i$ for each $i\in N$.
Thus $(\mathcal{Y},\succeq_R)$ is a lattice.

\subsection{Top trading cycles, the crawler and their properties}\label{subsec:crawler}

\textbf{Top trading cycles (TTC)}, introduced by \citet{shapley1974cores}, is denoted by $\TTC\colon\mathcal{R}^n\to\mathcal{X}$. 
Fix a preference profile $R\in\mathcal{R}^n$. 
TTC assigns objects through the following procedure:
\begin{enumerate}
    \item[1.] At each step, each remaining agent $i$ points to the agent $j$ whose object is $i$'s most-preferred remaining object. At least one directed cycle exists.
    \item[2.] Each agent $i$ in a cycle receives $j$'s object, where $j$ is the agent to whom $i$ points.
    \item[3.] All agents in cycles leave with their assigned objects. 
\end{enumerate}
The procedure repeats until no agent remains, and yields $\TTC(R)\in\mathcal{X}$.

Next, we introduce the \textbf{crawler} of \citet{bade2019matching},
denoted by $\CR\colon\mathcal{R}^n\to\mathcal{X}$.
For a nonempty set $O'\subseteq O$, let $\max_{R_i}O'$ denote
agent $i$'s most preferred object in $O'$.
Fix a preference profile $R\in\mathcal{R}^n$.
The crawler assigns objects through the following procedure.

\paragraph{The crawler:}
At Step $k$, let $N^k$ and $O^k$ denote the sets of remaining agents and objects, and let $o_i^k\in O^k$ denote the object
currently held by agent $i\in N^k$.
Initially, $N^1=N$, $O^1=O$, and $o_i^1=\omega_i$ for each $i\in N$.
\begin{enumerate}
    \item[1.] At each step, consider the remaining agents in increasing order
of their current holdings.
Select the first agent $i$ whose most preferred remaining
object is weakly below their current holding:
$\max_{R_i}O^k\le o_i^k$.
Equivalently, $i$ is the unique remaining agent satisfying $\max_{R_i}O^k\le o_i^k$ and 
$o_j^k<\max_{R_j}O^k$ for each  $j\in N^k$ with $o_j^k<o_i^k$. Assign $o'=\max_{R_i}O^k$ to $i$ and remove both $i$ and $o'$: $\CR_i(R)=o'$, $N^{k+1}=N^k\setminus\{i\}$, and 
$O^{k+1}=O^k\setminus\{o'\}$.
\item[2.] For each remaining agent $j$, define their next holding by
\begin{equation*}
o_j^{k+1}=
\begin{cases}
\min_{<}\{o\in O^{k+1}\colon o_j^k<o\}
 &\text{if }o'\le o_j^k<o_i^k,\\
o_j^k
 &\text{otherwise}.
\end{cases}
\end{equation*}
Here the minimum is taken with respect to $<$.
\item[3. ]If agents remain, begin the next step with the agent holding
the smallest remaining object.
\end{enumerate}
The procedure terminates after $n$ steps and yields $\CR(R)\in\mathcal{X}$.

Next, we use the following properties of rules.
\paragraph{Efficiency:}
A rule $f$ is \textit{efficient} if, for each $R\in\mathcal{R}^n$, there does not exist $x\in\mathcal{X}$ such that $x\succeq_Rf(R)$ and 
$x_j\mathbin{P_j}f_j(R)$ for some $j\in N$. 
\paragraph{Individual rationality:}
A rule $f$ is \textit{individually rational} if, for each
$R\in\mathcal{R}^n$ and each $i\in N$,
$f_i(R)\mathbin{R_i}\omega_i$.\footnote{\textit{Individual rationality} is also referred to as
\textit{the endowments lower bounds}.}
\paragraph{Strategy-proofness:}
A rule $f$ is \textit{strategy-proof} if, for each 
$R\in\mathcal{R}^n$, each $i\in N$, and each $R_i'\in\mathcal{R}$, $f_i(R)\mathbin{R_i}f_i(R_i',R_{-i})$.

\mbox{}\par
TTC satisfies \emph{efficiency}, \emph{individual rationality}, and \emph{strategy-proofness} (\citet{shapley1974cores,roth1982incentive}). On the unrestricted domain of strict preferences, it is the unique rule satisfying these properties (\citet{ma1994strategy}). However, on the single-peaked preference domain, the crawler also satisfies these properties (\citet{bade2019matching}).

To analyze the crawler, we consider an additional property, \emph{non-bossiness}, which requires that, if an agent changes their report without changing their own assigned object, each other agent's assigned object also remains unchanged.
\paragraph{Non-bossiness:}
A rule $f$ is \textit{non-bossy} if, for each 
$R\in\mathcal{R}^n$, each $i\in N$, and  each $R_i'\in\mathcal{R}$, $f_i(R_i',R_{-i})=f_i(R)$
implies $f(R_i',R_{-i})=f(R)$.

\mbox{}\par
\citet{tamura2022object} shows that the crawler is \emph{non-bossy}.

\subsection{The direct revelation game and Nash equilibrium}\label{subsec: ne}

Given a rule $f$, a \textbf{direct revelation mechanism} is a pair
$\Gamma^f=(\mathcal{R}^n,f)$.
For a true preference profile $R^0\in\mathcal{R}^n$,
the \textbf{direct revelation game} of $f$
is $(\Gamma^f,R^0)$.
A profile $R\in\mathcal{R}^n$ is a pure strategy \textbf{Nash equilibrium}
of $(\Gamma^f,R^0)$ if, for each $i\in N$
and each $R_i'\in\mathcal{R}$, $f_i(R)\mathbin{R_i^0}f_i(R_i',R_{-i})$.
We focus on the pure strategy Nash equilibria of $(\Gamma^f,R^0)$. 
Let $\mathbf{NE}(\Gamma^f,R^0)$ denote the set of all such profiles, and let
$\mathcal{X}\bigl(\NE(\Gamma^f,R^0)\bigr)$
denote the set of Nash equilibrium allocations:
$\mathcal{X}\bigl(\NE(\Gamma^f,R^0)\bigr)
=\{f(R):R\in\NE(\Gamma^f,R^0)\}$.

Hereafter, we analyze the direct revelation game of the crawler $(\Gamma^{\CR},R^0)$.

\section{Result}\label{sec:result}
We begin with an example illustrating the lattice structure of the crawler's Nash equilibrium allocations.
\begin{emp}[\textit{The lattice structure of the crawler}]\label{emp: lattice}
    Let $N=\{1,2,3,4\}$ and $\omega=(o_1,o_2,o_3,o_4)$, where $o_1<o_2<o_3<o_4$. 
    Consider the following preference profile $R^0$: 
    \begin{align*}
        R^0_1&\colon o_2,~ o_1,~ o_3,~ o_4,\\
        R^0_2&\colon o_1,~ o_2,~ o_3,~ o_4,\\
        R^0_3&\colon o_4,~ o_3,~ o_2,~ o_1,\\
        R^0_4&\colon o_3,~ o_4,~ o_2,~ o_1.
    \end{align*}
Then $\mathcal{X}\bigl(\NE(\Gamma^{\CR},R^0)\bigr)=\{(o_1,o_2,o_3,o_4),(o_2,o_1,o_3,o_4),(o_1,o_2,o_4,o_3),(o_2,o_1,o_4,o_3)\}$, where $\omega=(o_1,o_2,o_3,o_4)$ and $\CR(R^0)=(o_2,o_1,o_4,o_3)$.
Under the order of weak Pareto dominance $\succeq_{R^0}$, $\mathcal{X}\bigl(\NE(\Gamma^{\CR},R^0)\bigr)$ is a  complete lattice whose Hasse diagram is shown in Figure \ref{fig:crawler-ne-diamond}.
\begin{figure}[H]
\centering
\begin{tikzpicture}[
    x=1cm, y=1cm,
    every node/.style={align=center, inner sep=4pt},
    every path/.style={semithick}
]

\node (bottom) at (0,0)
    {$\omega$};

\node (left) at (-2.6,1.6)
    {$(o_2,o_1,o_3,o_4)$};

\node (right) at (2.6,1.6)
    {$(o_1,o_2,o_4,o_3)$};

\node (top) at (0,3.2)
    {$\textit{CR}(R^0)$};

\draw (bottom) -- (left);
\draw (bottom) -- (right);
\draw (left) -- (top);
\draw (right) -- (top);

\end{tikzpicture}
\caption{Hasse diagram for Example \ref{emp: lattice}.}
\label{fig:crawler-ne-diamond}
\end{figure}
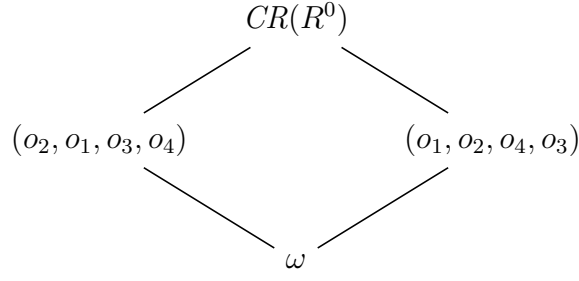
\end{emp}

The following theorem establishes this lattice structure for
each true preference profile.

\begin{thm}\label{thm:lattice}
For each true preference profile $R^0\in\mathcal{R}^n$,
$\bigl(\mathcal{X}\bigl(\NE(\Gamma^{\CR},R^0)\bigr),\succeq_{R^0}\!\bigr)$ is a complete lattice,
whose greatest element is the truthful allocation and least element is the initial endowment allocation.
Thus,
\begin{equation*}
    \bigvee \mathcal{X}\bigl(\NE(\Gamma^{\CR},R^0)\bigr)=\CR(R^0),
\qquad
\bigwedge \mathcal{X}\bigl(\NE(\Gamma^{\CR},R^0)\bigr)=\omega.
\end{equation*}
\end{thm}

The same conclusions need not hold for TTC, although it also
satisfies \emph{efficiency}, \emph{individual rationality}, and \emph{strategy-proofness}.
The following example illustrates this difference.

\begin{emp}[\textit{Failure of the lattice structure under TTC}]
\label{emp:ttc-nonlattice}
Let $N=\{1,2,3\}$ and
$\omega=(o_1,o_2,o_3)$, where $o_1<o_2<o_3$.
Consider the following preference profile $R^0$:
\begin{align*}
    R_1^0&\colon o_3,~o_2,~o_1,\\
    R_2^0&\colon o_3,~o_2,~o_1,\\
    R_3^0&\colon o_2,~o_1,~o_3.
\end{align*}
Then
$\mathcal{X}\bigl(\NE(\Gamma^{\TTC},R^0)\bigr)
=
\bigl\{
(o_1,o_2,o_3),~
(o_1,o_3,o_2),~
(o_3,o_2,o_1)
\bigr\}$,
where $\omega=(o_1,o_2,o_3)$ and $\TTC(R^0)=(o_1,o_3,o_2)$.
Under the order of weak Pareto dominance $\succeq_{R^0}$, $\mathcal{X}\bigl(\NE(\Gamma^{\TTC},R^0)\bigr)$ is not a lattice. 
The allocations $\TTC(R^0)$ and $(o_3,o_2,o_1)$ have no common
upper bound in $\mathcal{X}\bigl(\NE(\Gamma^{\TTC},R^0)\bigr)$.
Hence, their join does not exist, and this set is not a lattice.
Moreover, the truthful allocation does not weakly Pareto dominate
$(o_3,o_2,o_1)$.
Its Hasse diagram is shown in Figure \ref{fig:ttc-ne-nonlattice}.

\begin{figure}[H]
\centering
\begin{tikzpicture}[
    x=1cm, y=1cm,
    every node/.style={align=center, inner sep=4pt},
    every path/.style={semithick}
]
\node (bottom) at (0,0)
    {$\omega$};

\node (left) at (-2.6,1.6)
    {$\textit{TTC}(R^0)$};

\node (right) at (2.6,1.6)
    {$(o_3,o_2,o_1)$};

\draw (bottom) -- (left);
\draw (bottom) -- (right);
\end{tikzpicture}
\caption{Hasse diagram for Example \ref{emp:ttc-nonlattice}.}
\label{fig:ttc-ne-nonlattice}
\end{figure}
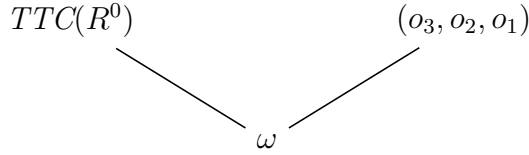
\end{emp}

\section{Concluding remarks}\label{sec:conclusion}
This paper studies the Nash equilibrium allocations of the crawler in object reallocation problems with single-peaked preferences.
We show that these allocations form a complete lattice under weak Pareto dominance, with the truthful allocation as its greatest element and the initial endowment allocation as its least element.
Thus, the truthful allocation weakly Pareto dominates each equilibrium allocation.
The same conclusions need not hold for TTC.

An important direction for future research is to examine the relation between lattice structure and obviously strategy-proof implementation.
\citet{bochet2020dominance} establish the lattice structure of Nash equilibrium allocations for sequential allotment rules in infinitely divisible obeect allocation problems.
\citet{arribillaga2023all} show that sequential allotment rules  admit obviously strategy-proof implementations.\footnote{They study a discrete division problem to avoid technical
difficulties.}
The crawler also admits such an implementation \citep{bade2019matching}.
These results motivate a comparison of the conditions that yield lattice structure and those that permit obviously strategy-proof implementation.

However, obviously strategy-proof implementability alone does not ensure lattice structure.
On the single-peaked preference domain, TTC admits an obviously strategy-proof implementation when there are at most three agents, but Example~\ref{emp:ttc-nonlattice} shows that its Nash equilibrium allocations need not form a lattice.
Our main direction for future research is therefore to explore the relation between the lattice structure of Nash equilibrium allocations and obviously strategy-proof implementability.

\appendix
\section*{Appendix: Proofs}\label{app:proofs}
We first establish two preliminary lemmas.

\begin{lem}\label{lem: single-peaked}
Fix $R_i\in\mathcal{R}$ and $o,o',o''\in O$ with $o<o'<o''$. Then $o\mathbin{P_i}o'$ and $o''\mathbin{P_i}o'$ cannot both hold.  
\end{lem}

\begin{prf}
    If $o'\le p_i$, by single-peakedness, $o'\mathbin{P_i}o$. If $p_i< o'$, by single-peakedness, $o'\mathbin{P_i}o''$.
    At least one of these cases holds. Since $P_i$ is asymmetric, $o\mathbin{P_i}o'$ and $o''\mathbin{P_i}o'$ cannot both hold.  \qed
\end{prf}

For each $i\in N$ and $o\in O$, fix $R_i^o\in\mathcal{R}$ such that
$o\mathbin{R_i^o}o'$ for each $o'\in O$.

\begin{lem}\label{lem: CR=h}
    Fix $i\in N$, $o\in O$, and $R_{-i}\in\mathcal{R}^{n-1}$. There exists $R'_i\in\mathcal{R}$ such that $\CR_i(R'_i,R_{-i})=o$ if and only if $\CR_i(R^o_i,R_{-i})=o$. 
\end{lem}

\begin{prf}
    Suppose that there exists $R'_i\in\mathcal{R}$ such that $\CR_i(R'_i,R_{-i})=o$.
    By \emph{strategy-proofness},
    \begin{equation*}
        \CR_i(R_i^o,R_{-i})\mathbin{R_i^o}\CR_i(R'_i,R_{-i})=o.
    \end{equation*}
    Since $o$ is the most preferred object under $R_i^o$, $\CR_i(R_i^o,R_{-i})=o$. 
    Conversely, choose $R'_i=R_i^o$.\qed
\end{prf}

For $x\in\mathcal{X}$, let $R^x=(R_i^{x_i})_{i\in N}$.
Since $x$ assigns each agent their most preferred object under $R^x$,
it weakly Pareto dominates each allocation at $R^x$.
By \emph{efficiency}, $\CR(R^x)=x$. 
Define $O_i(x)=\{\CR_i(R'_i,R^x_{-i})\colon R'_i\in\mathcal{R}\}$ and $N_x(o)=\{i\in N\colon o\in O_i(x)\}$.
Thus $O_i(x)$ consists of the objects $i$ can obtain by changing only their report while the other
reports remain $R^x_{-i}$.
The set $N_x(o)$ consists of agents who can obtain $o$ in this sense.

\begin{lem}\label{lem: NE}
 Fix a true preference profile $R^0\in\mathcal{R}^n$ and $x\in\mathcal{X}$. Then $x\in\mathcal{X}\bigl(\NE(\Gamma^{\CR},R^0)\bigr)$ if and only if $x_i\mathbin{R_i^0} o$ for each $i\in N$ and each $o\in O_i(x)$. 
\end{lem}

\begin{prf}
    Suppose $x\in\mathcal{X}\bigl(\NE(\Gamma^{\CR},R^0)\bigr)$.
    Choose $R'\in\NE(\Gamma^{\CR},R^0)$ with $\CR(R')=x$, and fix $i\in N$. 
    By the definition of Nash equilibrium, $x_i\mathbin{R_i^0}\CR_i(R_i^0,R'_{-i})$. 
    By \emph{strategy-proofness}, 
    $\CR_i(R_i^0,R'_{-i})\mathbin{R_i^0}x_i$.
    Thus $\CR_i(R_i^0,R'_{-i})=x_i$. 
    Since $\CR_i(R')=\CR_i(R_i^0,R'_{-i})=x_i$, by \emph{non-bossiness}, $\CR(R_i^0,R'_{-i})=x$.
    Fix $j\ne i$ and change $j$'s report in $(R_i^0,R'_{-i})$ to $R_j^{x_j}$.
    By \emph{strategy-proofness} under $R_j^{x_j}$,
    \begin{equation*}
    \CR_j(R_i^0,R_j^{x_j},R'_{-\{i,j\}})
    \mathbin{R_j^{x_j}}\CR_j(R_i^0,R'_{-i})=x_j.
    \end{equation*}
    Since $x_j$ is the most preferred object under $R_j^{x_j}$,
    $\CR_j(R_i^0,R_j^{x_j},R'_{-\{i,j\}})=x_j$.
    By \emph{non-bossiness}, $\CR(R_i^0,R_j^{x_j},R'_{-\{i,j\}})=x$.
    Repeating this argument for the remaining agents other than $i$, one at a time,
    gives $\CR(R_i^0,R_{-i}^x)=x$.
    By \emph{strategy-proofness} under $R_i^0$, for each $R_i'\in\mathcal{R}$, we have $x_i\mathbin{R_i^0}\CR_i(R'_i,R_{-i}^x)$. 
    By the definition of $O_i(x)$, this implies $x_i\mathbin{R_i^0} o$ for each $o\in O_i(x)$. 

    Conversely, suppose that $x_i\mathbin{R_i^0} o$ for each $i\in N$ and each $o\in O_i(x)$.
    Since $\CR(R^x)=x$ and $O_i(x)=\{\CR_i(R_i',R_{-i}^x)\colon R'_i\in\mathcal{R}\}$, by the assumption, for each $i\in N$ and each $R'_i\in \mathcal{R}$, we have
    \begin{equation*}
        \CR_i(R^x)=x_i\mathbin{R_i^0}\CR_i(R'_i,R_{-i}^x).
    \end{equation*}
    This implies that $R^x\in\NE(\Gamma^{\CR},R^0)$ and $x\in\mathcal{X}\bigl(\NE(\Gamma^{\CR},R^0)\bigr)$.\qed
\end{prf}

For $g,h\in O$ and $i,j\in N$, define 
\begin{align*}
    (g,h]_O&=\{o\in O\colon g<o\le h\}, &
    (g,h)_O&=\{o\in O\colon g<o< h\}, \\
    (i,j)_N&=\{\ell\in N\colon \omega_i<\omega_{\ell}<\omega_j\}, & 
    [i,j]_N&=\{\ell\in N\colon \omega_i\le\omega_{\ell}\le\omega_j\},\\
    (i,j]_N&=\{\ell\in N\colon \omega_i<\omega_\ell\le\omega_j\},&
    [i,j)_N&=\{\ell\in N\colon \omega_i\le\omega_\ell<\omega_j\}.
\end{align*}
Since $\omega\in\mathcal{X}$, for $\omega_i\le\omega_j$,
\begin{equation*}
\#(i,j)_N=\#(\omega_i,\omega_j)_O,\qquad
\#(i,j]_N=\#[i,j)_N=\#(\omega_i,\omega_j]_O.
\end{equation*}

\begin{lem}\label{lem: cr_i=o iff}
Fix $x\in\mathcal{X}$, $o\in O$, $j=x_o$, and $i\in N\setminus\{j\}$.
Let $\hat{R}=(R_i^o,R_{-i}^x)\in\mathcal{R}^n$. 
If $\omega_i<\omega_j$, $\CR_i(\hat{R})=o$ if and only if 
\begin{equation*}
    \#\{\ell\in(i,j)_N\colon x_{\ell}<o\}\ge \#(\omega_i,o]_O.
\end{equation*}
If $\omega_j<\omega_i$, $\CR_i(\hat{R})=o$ if and only if 
\begin{equation*}
    \#\{\ell\in(j,i)_N\colon x_{\ell}<o\}< \#(\omega_j,o]_O.
\end{equation*}
\end{lem}

\begin{prf}
    Under $\hat{R}\in\mathcal{R}^n$, $o$ is the most-preferred object for $i$ and $j$. $x_\ell$ is the most-preferred object for each other agent $\ell$. These objects are all distinct and differ from $o$.   

    We first show that the following statements hold at the start of each Step $k$ before $o$ is assigned: $i,j\in N^k$, $o\in O^k$,  $\max_{\hat{R}_i} O^k=\max_{\hat{R}_j}O^k=o$, and for each $\ell\in N^k\setminus\{i,j\}$, $x_\ell\in O^k$ and $\max_{\hat{R}_\ell}O^k=x_\ell$. 
    These statements hold at Step 1 by the definition of $\hat{R}$.
    Suppose they hold at the start of Step $k$.
    If $i$ or $j$ is selected, that agent receives $o$. If $\ell\notin\{i,j\}$ is selected, since $x_\ell$  is the most-preferred object for $\ell$, $\CR_\ell(\hat{R})=x_\ell\ne o$, $N^{k+1}=N^k\setminus\{\ell\}$, and $O^{k+1}=O^k\setminus\{x_\ell\}$.
    Thus $i,j\in N^{k+1}$ and $o\in O^{k+1}$.
    For $m\in N^{k+1}\setminus\{i,j\}$, by the induction hypothesis, $x_m\in O^k$.
    Hence $x_m\in O^k\setminus\{x_\ell\}=O^{k+1}$.
    Each remaining agent's first choice is therefore still available, so the statements follow by induction.
    
    The crawler assigns each object in finitely many steps. 
    The statements above show that only $i$ or $j$ can receive $o$. 
     Thus $\CR_\ell(\hat{R})\ne o$ for  each $\ell\notin\{i,j\}$. 
     Therefore $\CR_i(\hat{R})=o$ if and only if $\CR_j(\hat{R})\ne o$.
     
     Suppose that $\omega_i<\omega_j$.
     At each Step $k$ before $o$ is assigned, define
     \begin{align*}
         \hat{N}^k&=\{m\in N^k\colon o_i^k<o_m^k<o_j^k,x_m<o\},\\
         \hat{O}^k&=\{o'\in O^k\colon o_i^k< o'\le o\}.
     \end{align*}
     Note that, since the remaining agents stay in the same order, 
     \begin{equation*}
         \hat{N}^k=\{m\in N^k\colon\omega_i<\omega_m<\omega_j,x_m<o\}.
     \end{equation*}
     Whenever $\ell\notin\{i,j\}$ leaves, it follows that $\hat{N}^{k+1}=\hat{N}^k\setminus\{\ell\}$.
     In particular $\hat{N}^{k+1}=\hat{N}^k$ if $\ell\notin\hat{N}^k$.

     We prove the following claim by induction on the number $m=\#N^k$ of the remaining agents: 
     at each step before $o$ is assigned, $\CR_i(\hat{R})=o$ if and only if $\#\hat{N}^k\ge\#\hat{O}^k$. 
     If $m=2$, then $N^k=\{i,j\}$, $O^k=\{o_i^k,o_j^k\}$, and $\hat{N}^k=\emptyset$. 
     If $o=o_i^k$, $\CR_i(\hat{R})=o$ and $\#\hat{N}^k=\#\hat{O}^k=0$. 
     If $o=o_j^k$, $\CR_j(\hat{R})=o$, so $\CR_i(\hat{R})\ne o$ and $\#\hat{N}^k=0<1=\#\hat{O}^k$. 

     Let $m> 2$, and suppose the claim holds whenever $m-1$ agents remain and $o$ is unassigned.  
     Consider a step with $m$ remaining agents.

     \medskip
     \noindent\textsc{Case 1:} Agent $i$ leaves.
     By the definition of the crawler and the statements above, $\CR_i(\hat{R})=o\le o_i^k$.
     Hence $\hat{O}^k=\emptyset$ and $\#\hat{N}^k\ge\#\hat{O}^k=0$.
     
     \medskip
     \noindent\textsc{Case 2:} Agent $j$ leaves.
     By the definition of the crawler, $o_i^k<o\le o_j^k$, and for each $\ell\in \hat{N}^k$, 
     \begin{equation*}
    o_i^k<o_\ell^k<\max_{\hat{R}_\ell}O^k=x_\ell<o.
     \end{equation*}
     Hence $\{o_\ell^k\in O^k\colon \ell\in\hat{N}^k\}\subseteq\hat{O}^k\setminus\{o\}$.
     Since $o_\ell^k\ne o_m^k$ for distinct $\ell,m\in N^k$, and $o\in\hat{O}^k$,
     \begin{equation*}
         \#\hat{N}^k=\#\{o_\ell^k\in O^k\colon \ell\in\hat{N}^k\}\le\#(\hat{O}^k\setminus\{o\})=\#\hat{O}^k-1<\#\hat{O}^k.
     \end{equation*}
     Also, $\CR_j(\hat{R})=o$, so $\CR_i(\hat{R})\ne o$. 

     \medskip\noindent\textsc{Case 3:} An agent $\ell\notin\{i,j\}$ leaves.
     Here $\CR_\ell(\hat{R})=x_\ell\ne o$, so $N^{k+1}=N^k\setminus\{\ell\}$ and $O^{k+1}=O^k\setminus\{x_\ell\}$.
     In particular, $o$ remains unassigned and $\# N^{k+1}=m-1$. 
     We show that $\#\hat{N}^k\ge\#\hat{O}^k$ and $\#\hat{N}^{k+1}\ge\#\hat{O}^{k+1}$ either both hold or both fail. 

     \medskip
     \noindent\textsc{Case 3.1:} $o\le o_j^k$.
     Since $\max_{\hat{R}_j}O^k=o\le o_j^k$ and $\ell\ne j$ is selected, by the definition of the crawler, $x_\ell\le o_\ell^k<o_j^k$.
     \begin{enumerate}
         \item[(i)] If $o_\ell^k<o_i^k$, then $x_\ell<o_i^k$ and $o_i^{k+1}=o_i^k$. 
         Since $\ell\notin\hat{N}^k$ and $x_\ell\notin\hat{O}^k$, $\hat{N}^k=\hat{N}^{k+1}$ and $\hat{O}^k=\hat{O}^{k+1}$.
         \item[(ii)] If $o_i^k<o_\ell^k$ and $o<x_\ell$, then $o_i^k<o$. Thus $o_i^k<o<x_\ell\le o_\ell^k<o_j^k$.
         Since $o_i^{k+1}=o_i^k$, $\ell\notin \hat{N}^k$, and $x_\ell\notin\hat{O}^k$, we have $\hat{N}^k=\hat{N}^{k+1}$ and $\hat{O}^k=\hat{O}^{k+1}$.
         \item[(iii)] If $o_i^k<o_\ell^k$ and $x_\ell<o$, then $o_i^k<o$ and $\ell\in\hat{N}^k$. 
         Thus $\hat{N}^{k+1}=\hat{N}^k\setminus\{\ell\}$ and $\#\hat{N}^{k+1}=\#\hat{N}^k-1$.
         If $o_i^k<x_\ell<o$, then $x_\ell\in\hat{O}^k$ and $o_i^{k+1}=o_i^k$, implying $\hat{O}^{k+1}=\hat{O}^k\setminus\{x_\ell\}$ and $\#\hat{O}^{k+1}=\#\hat{O}^k-1$.
         If $x_\ell\le o_i^k$, then $x_\ell\notin \hat{O}^k$. 
         Since $o_i^k<o$ and $o\in O^{k+1}$, by the definition of $o_i^{k+1}$,
         \begin{equation*}
             o_i^{k+1}=\min_<\{o'\in O^{k+1}\colon o_i^k<o'\}\le o.
         \end{equation*}
         Since $x_\ell\notin\hat{O}^k$ and $o_i^{k+1}=\min_<\hat{O}^k$, by the definition of $\hat{O}^{k+1}$, we have $o_i^{k+1}\in\hat{O}^k$, $\hat{O}^{k+1}=\hat{O}^k\setminus\{o_i^{k+1}\}$, and $\#\hat{O}^{k+1}=\#\hat{O}^k-1$.
     \end{enumerate}
     In each case, $\#\hat{N}^{k+1}-\#\hat{O}^{k+1}=\#\hat{N}^k-\#\hat{O}^k$.
     Thus $\#\hat{N}^{k+1}\ge \#\hat{O}^{k+1}$ if and only if $\#\hat{N}^{k}\ge \#\hat{O}^{k}$.

     \medskip
     \noindent\textsc{Case 3.2:} $o_j^k<o$.
     For each $\ell\in\hat{N}^k$, $o_i^k<o_\ell^k<o_j^k<o$.
     Hence $\{o_\ell^k\in O^k\colon \ell\in\hat{N}^k\}\subseteq\hat{O}^k\setminus\{o\}$ and $\#\hat{N}^k=\#\{o_\ell^k\in O^k\colon \ell\in\hat{N}^k\}\le \#\hat{O}^k-1$.
     Since $o\in O^{k+1}$, by the definitions of $o_i^{k+1}$ and $o_j^{k+1}$, $o_i^{k+1}<o_j^{k+1}\le o$. 
     Similarly, at Step $k+1$, 
     $\{o_\ell^{k+1}\in O^{k+1}\colon \ell\in\hat{N}^{k+1}\}\subseteq\hat{O}^{k+1}\setminus\{o\}$ and $\#\hat{N}^{k+1}=\#\{o_\ell^{k+1}\in O^{k+1}\colon \ell\in\hat{N}^{k+1}\}\le\#\hat{O}^{k+1}-1$.
     In \textsc{Case 3}, apply the induction hypothesis at Step $k+1$. By the preceding inequalities, $\CR_i(\hat{R})=o$ if and only if $\#\hat{N}^k\ge\#\hat{O}^k$. This completes the induction. 
     At Step $1$, $N^1=N$, $O^1=O$, and $o^1_\ell=\omega_\ell$, so 
     \begin{align*}
         \hat{N}^1 &= \{\ell\in N\colon o_i^1<o_\ell^1<o_j^1,x_\ell<o\}\\
         &= \{\ell\in N\colon \omega_i<\omega_\ell<\omega_j,x_\ell<o\}\\
         &=\{\ell\in(i,j)_N\colon x_\ell<o\},\\
         \hat{O}^1 &= \{o'\in O\colon o_i^1<o'\le o\}\\
         &= \{o'\in O\colon \omega_i< o'\le o\}\\
         &=(\omega_i,o]_O.
     \end{align*}
     Hence $\#\hat{N}^1=\#\{\ell\in(i,j)_N\colon x_\ell<o\}$ and  $\#\hat{O}^1=\#(\omega_i,o]_O$. 

     If $\omega_j<\omega_i$, exchange the roles of $i$ and $j$.
     Since $\max_{\hat{R}_i}O=\max_{\hat{R}_j}O=o$, the preceding argument gives $\CR_j(\hat{R})=o$ if and only if 
     \begin{equation*}
          \#\{\ell\in(j,i)_N\colon x_{\ell}<o\}\ge \#(\omega_j,o]_O.
     \end{equation*}
     Since $\CR_i(\hat{R})=o$ exactly when $\CR_j(\hat{R})\ne o$, we obtain $\CR_i(\hat{R})=o$ if and only if 
     \begin{equation*}
          \#\{\ell\in(j,i)_N\colon x_{\ell}<o\}< \#(\omega_j,o]_O.
     \end{equation*}\qed
\end{prf}

    \begin{lem}\label{lem: U,L}
               Let $x\in\mathcal{X}$, $o\in O$, and $j=x_o$.
If $\omega_j<o$, let $U$ be the agent with the smallest endowment among
those $i$ satisfying $\omega_j<\omega_i$ and 
\begin{equation*}
\#\{\ell\in N:\omega_j<\omega_\ell\le\omega_i,\ x_\ell<o\}
=\#(\omega_j,o]_O.
\end{equation*}
If $o<\omega_j$, let $L$ be the agent with the largest endowment among
those $i$ satisfying $\omega_i<\omega_j$ and
\begin{equation*}
\#\{\ell\in N:\omega_i\le\omega_\ell<\omega_j,\ o<x_\ell\}
=\#(o,\omega_j]_O.
\end{equation*}
The agents $U$ and $L$ exist and are unique in their respective cases.
Moreover, $o\le\omega_U$ if $\omega_j<o$, and
$\omega_L\le o$ if $o<\omega_j$.
    \end{lem}

\begin{prf}
\mbox{}\\
\medskip
\noindent\textsc{Case 1:} $\omega_j<o$.
We first show that $U$ is well-defined.
Since $x_j=o$,
\begin{equation*}
\{\ell\in N:x_\ell<o\}
=\{\ell\in N:\omega_\ell<\omega_j,\ x_\ell<o\}\cup\{\ell\in N:\omega_j<\omega_\ell,\ x_\ell<o\}.
\end{equation*}
The two sets are disjoint. Since $x\in\mathcal{X}$,
$\#\{\ell\in N:x_\ell<o\}=\#\{o'\in O:o'<o\}$.
Thus
\begin{align*}
\#\{\ell\in N:\omega_j<\omega_\ell,\ x_\ell<o\}
&=\#\{\ell\in N:x_\ell<o\}
       -\#\{\ell\in N:\omega_\ell<\omega_j,\ x_\ell<o\}\\
&=\#\{o'\in O:o'<o\}
       -\#\{\ell\in N:\omega_\ell<\omega_j,\ x_\ell<o\}\\
&\ge\#\{o'\in O:o'<o\}
       -\#\{\ell\in N:\omega_\ell<\omega_j\}\\
&=\#\{o'\in O:o'<o\}
       -\#\{o'\in O:o'<\omega_j\}.
\end{align*}

Since $\omega_j<o$,
\begin{equation*}
\{o'\in O:o'<o\}
=\{o'\in O:o'<\omega_j\}\cup\{\omega_j\}\cup(\omega_j,o)_O.
\end{equation*}
Thus
\begin{equation*}
\#\{o'\in O:o'<o\}-\#\{o'\in O:o'<\omega_j\}
=1+\#(\omega_j,o)_O=\#(\omega_j,o]_O,
\end{equation*}
It follows that $\#\{\ell\in N:\omega_j<\omega_\ell,\ x_\ell<o\}\ge\#(\omega_j,o]_O>0$.

Start with $i=j$, for which $\{\ell\in N:\omega_j<\omega_\ell\le\omega_j,\ x_\ell<o\}=\emptyset$.
Then consider the agents $i$ with $\omega_j<\omega_i$ in increasing order
of $\omega_i$. Including $i$ adds one agent to the set being counted
if $x_i<o$, and adds none if $o\le x_i$:
\begin{equation*}
\#\{\ell\in N:\omega_j<\omega_\ell\le\omega_i,\ x_\ell<o\}-\#\{\ell\in N:\omega_j<\omega_\ell<\omega_i,\ x_\ell<o\}
=\begin{cases}
1&\text{if $x_i<o$,}\\
0&\text{if $o\le x_i$.}
\end{cases}
\end{equation*}
After all agents with $\omega_j<\omega_i$ have been considered, the number is
\begin{equation*}
\#\{\ell\in N:\omega_j<\omega_\ell,\ x_\ell<o\}\ge\#(\omega_j,o]_O>0.
\end{equation*}
Starting from zero and increasing only by zero or one, the number must equal
$\#(\omega_j,o]_O$ for some $i$.
Thus $U$ exists: it is the agent with the smallest $\omega_i$
for which this equality holds.
Since $\omega_\ell\ne\omega_m$ for $\ell\ne m$, $U$ is unique.

By the definition of $U$ and since $\omega\in\mathcal{X}$,
\begin{align*}
\#(\omega_j,o]_O
&=\#\{\ell\in N:\omega_j<\omega_\ell\le\omega_U,\ x_\ell<o\}\\
&\le\#\{\ell\in N:\omega_j<\omega_\ell\le\omega_U\}\\
&=\#(\omega_j,\omega_U]_O.
\end{align*}
If $\omega_U<o$, then $(\omega_j,o]_O
=(\omega_j,\omega_U]_O\cup(\omega_U,o]_O$.
Since $o\in(\omega_U,o]_O$, we obtain
\begin{equation*}
\#(\omega_j,o]_O=\#(\omega_j,\omega_U]_O+\#(\omega_U,o]_O>\#(\omega_j,\omega_U]_O,
\end{equation*}
contrary to $\#(\omega_j,o]_O\le\#(\omega_j,\omega_U]_O$.
Hence $o\le\omega_U$.

\medskip
\noindent\textsc{Case 2:} $o<\omega_j$.
We first show that $L$ is well-defined.
Since $x_j=o$,
\begin{equation*}
\{\ell\in N:o<x_\ell\}
=\{\ell\in N:\omega_\ell<\omega_j,\ o<x_\ell\}
\cup\{\ell\in N:\omega_j<\omega_\ell,\ o<x_\ell\}.
\end{equation*}
The two sets are disjoint. Since $x\in\mathcal{X}$,
$\#\{\ell\in N:o<x_\ell\}=\#\{o'\in O:o<o'\}$.
Thus 
\begin{align*}
\#\{\ell\in N:\omega_\ell<\omega_j,\ o<x_\ell\}
&=\#\{\ell\in N:o<x_\ell\}
       -\#\{\ell\in N:\omega_j<\omega_\ell,\ o<x_\ell\}\\
&=\#\{o'\in O:o<o'\}
       -\#\{\ell\in N:\omega_j<\omega_\ell,\ o<x_\ell\}\\
&\ge\#\{o'\in O:o<o'\}
       -\#\{\ell\in N:\omega_j<\omega_\ell\}\\
&=\#\{o'\in O:o<o'\}
       -\#\{o'\in O:\omega_j<o'\}.
\end{align*}
Since $\{o'\in O:o<o'\}
=(o,\omega_j]_O\cup\{o'\in O:\omega_j<o'\}$,
we have
\begin{equation*}
\#\{o'\in O:o<o'\}-\#\{o'\in O:\omega_j<o'\}
=\#(o,\omega_j]_O.
\end{equation*}
Thus $\#\{\ell\in N:\omega_\ell<\omega_j,\ o<x_\ell\}\ge\#(o,\omega_j]_O>0.$

Start with $i=j$, for which $\{\ell\in N:\omega_j\le\omega_\ell<\omega_j,\ o<x_\ell\}=\emptyset$.
Then consider the agents $i$ with $\omega_i<\omega_j$ in decreasing order
of $\omega_i$. Including $i$ adds one agent to the set being counted
if $o<x_i$, and adds none if $x_i\le o$:
\begin{equation*}
\#\{\ell\in N:\omega_i\le\omega_\ell<\omega_j,\ o<x_\ell\}-\#\{\ell\in N:\omega_i<\omega_\ell<\omega_j,\ o<x_\ell\}
=\begin{cases}
1&\text{if $o<x_i$,}\\
0&\text{if $x_i\le o$.}
\end{cases}
\end{equation*}
After all agents with $\omega_i<\omega_j$ have been considered, the number is
\begin{equation*}
\#\{\ell\in N:\omega_\ell<\omega_j,\ o<x_\ell\}\ge\#(o,\omega_j]_O>0.
\end{equation*}
Starting from zero and increasing only by zero or one, the number must equal
$\#(o,\omega_j]_O$ for some $i$.
Thus $L$ exists: it is the agent with the largest $\omega_i$
for which this equality holds.
Since $\omega_\ell\ne\omega_m$ for $\ell\ne m$, $L$ is unique.

By the definition of $L$ and since $\omega\in\mathcal{X}$,
\begin{align*}
\#(o,\omega_j]_O
&=\#\{\ell\in N:\omega_L\le\omega_\ell<\omega_j,\ o<x_\ell\}\\
&\le\#\{\ell\in N:\omega_L\le\omega_\ell<\omega_j\}\\
&=\#\{o'\in O:\omega_L\le o'<\omega_j\}\\
&=1+\#(\omega_L,\omega_j)_O\\
&=\#(\omega_L,\omega_j]_O.
\end{align*}
If $o<\omega_L$, then $(o,\omega_j]_O=(o,\omega_L]_O\cup(\omega_L,\omega_j]_O$.
Since $\omega_L\in(o,\omega_L]_O$, we obtain
\begin{equation*}
\#(o,\omega_j]_O=\#(o,\omega_L]_O+\#(\omega_L,\omega_j]_O>\#(\omega_L,\omega_j]_O,
\end{equation*}
contrary to $\#(o,\omega_j]_O\le\#(\omega_L,\omega_j]_O$.
Hence $\omega_L\le o$.\qed
\end{prf}

    \begin{lem}\label{lem: N_x(o)=}
        Let $x\in\mathcal{X}$, $o\in O$, and $j=x_o$. Let $U$ and $L$ be as in Lemma~\ref{lem: U,L} in their respective cases. Then
\begin{equation*}
N_x(o)=
\begin{cases}
[j,U]_N&\text{if $\omega_j<o$,}\\
[L,j]_N&\text{if $o<\omega_j$,}\\
\{j\}&\text{if $\omega_j=o$.}
\end{cases}
\end{equation*}
    \end{lem}
    \begin{prf}
        Since $\CR_j(R^x)=x_j=o$, we have $j\in N_x(o)$.

\medskip
\noindent\textsc{Case 1:} $\omega_j<o$.
By the definition of $U$ in Lemma~\ref{lem: U,L}, $x_U<o$ and
\begin{equation*}
\{\ell\in N:\omega_j<\omega_\ell\le\omega_U,\ x_\ell<o\}=\{\ell\in(j,U)_N:x_\ell<o\}\cup\{U\}.
\end{equation*}
Thus, $\#\{\ell\in(j,U)_N:x_\ell<o\}=\#(\omega_j,o]_O-1$.

We show that $[j,U]_N\subseteq N_x(o)$.
Let $i\in[j,U]_N$. If $i=j$, then $i\in N_x(o)$.
Suppose $i\ne j$. Then $\omega_j<\omega_i\le\omega_U$.
By $(j,i)_N\subseteq(j,U)_N$,
\begin{equation*}
\#\{\ell\in(j,i)_N:x_\ell<o\}\le\#\{\ell\in(j,U)_N:x_\ell<o\}=\#(\omega_j,o]_O-1<\#(\omega_j,o]_O.
\end{equation*}
Lemmas~\ref{lem: CR=h} and \ref{lem: cr_i=o iff} imply
$i\in N_x(o)$.

We next show that $N_x(o)\subseteq[j,U]_N$.
Let $i\in N_x(o)$.
If $\omega_i<\omega_j$, by $\omega\in\mathcal{X}$, 
$(\omega_i,\omega_j]_O
=(\omega_i,\omega_j)_O\cup\{\omega_j\}$, and 
$(\omega_i,o]_O
=(\omega_i,\omega_j]_O\cup(\omega_j,o]_O$,
we have
\begin{align*}
\#\{\ell\in(i,j)_N:x_\ell<o\}
&\le\#(i,j)_N\\
&=\#(\omega_i,\omega_j)_O\\
&=\#(\omega_i,\omega_j]_O-1\\
&<\#(\omega_i,o]_O,
\end{align*}
which contradicts Lemma~\ref{lem: cr_i=o iff}.
If $\omega_U<\omega_i$, then
\begin{equation*}
\{\ell\in N:\omega_j<\omega_\ell\le\omega_U,\ x_\ell<o\}\subseteq\{\ell\in(j,i)_N:x_\ell<o\},
\end{equation*}
and hence
\begin{align*}
\#\{\ell\in(j,i)_N:x_\ell<o\}
&\ge\#\{\ell\in N:\omega_j<\omega_\ell\le\omega_U,\ x_\ell<o\}\\
&=\#(\omega_j,o]_O.
\end{align*}
This again contradicts Lemma~\ref{lem: cr_i=o iff}.
Therefore $\omega_j\le\omega_i\le\omega_U$, so $i\in[j,U]_N$.

Therefore $N_x(o)=[j,U]_N$.

\medskip
\noindent\textsc{Case 2:} $o<\omega_j$.
By the definition of $L$ in Lemma~\ref{lem: U,L}, $o<x_L$ and
\begin{equation*}
\{\ell\in N:\omega_L\le\omega_\ell<\omega_j,\ o<x_\ell\}=\{L\}\cup\{\ell\in(L,j)_N:o<x_\ell\}.
\end{equation*}
Therefore $\#\{\ell\in(L,j)_N:o<x_\ell\}=\#(o,\omega_j]_O-1$.
For $\omega_i<o$, we will use the following identity.
Since $x_o=j\notin(i,j)_N$,
\begin{equation*}
(i,j)_N
=\{\ell\in(i,j)_N:x_\ell<o\}
 \cup\{\ell\in(i,j)_N:o<x_\ell\}.
\end{equation*}
Since $\omega\in\mathcal{X}$ and $\omega_i<o<\omega_j$,
\begin{align*}
\#(i,j)_N
&=\#(\omega_i,\omega_j)_O\\
&=\#(\omega_i,\omega_j]_O-1\\
&=\#(\omega_i,o]_O+\#(o,\omega_j]_O-1.
\end{align*}
Hence
\begin{equation*}
\#\{\ell\in(i,j)_N:x_\ell<o\}=\#(\omega_i,o]_O+\#(o,\omega_j]_O-1-\#\{\ell\in(i,j)_N:o<x_\ell\}.
\end{equation*}

We show that $[L,j]_N\subseteq N_x(o)$.
Let $i\in[L,j]_N$. If $i=j$, then $i\in N_x(o)$.
Suppose $i\ne j$. Then $\omega_L\le\omega_i<\omega_j$.
By $(i,j)_N\subseteq(L,j)_N$,
\begin{align*}
\#\{\ell\in(i,j)_N:o<x_\ell\}
&\le\#\{\ell\in(L,j)_N:o<x_\ell\}\\
&=\#(o,\omega_j]_O-1<\#(o,\omega_j]_O.
\end{align*}
If $\omega_i<o$, then
\begin{align*}
\#\{\ell\in(i,j)_N:x_\ell<o\}
&\ge\#(\omega_i,o]_O+\#(o,\omega_j]_O-1
     -\bigl(\#(o,\omega_j]_O-1\bigr)\\
&=\#(\omega_i,o]_O.
\end{align*}
If $o\le\omega_i$, then
\begin{equation*}
\#\{\ell\in(i,j)_N:x_\ell<o\}\ge0=\#(\omega_i,o]_O.
\end{equation*}
In either case, Lemmas~\ref{lem: CR=h} and \ref{lem: cr_i=o iff}
imply $i\in N_x(o)$.

We next show that $N_x(o)\subseteq[L,j]_N$.
Let $i\in N_x(o)$.
If $\omega_j<\omega_i$, then
\begin{equation*}
\#\{\ell\in(j,i)_N:x_\ell<o\}\ge0=\#(\omega_j,o]_O,
\end{equation*}
which contradicts Lemma~\ref{lem: cr_i=o iff}.
If $\omega_i<\omega_L$, then
\begin{equation*}
\{\ell\in N:\omega_L\le\omega_\ell<\omega_j,\ o<x_\ell\}\subseteq\{\ell\in(i,j)_N:o<x_\ell\},
\end{equation*}
so
\begin{align*}
\#\{\ell\in(i,j)_N:o<x_\ell\}
&\ge\#\{\ell\in N:\omega_L\le\omega_\ell<\omega_j,\ o<x_\ell\}\\
&=\#(o,\omega_j]_O.
\end{align*}
Since $\omega_i<\omega_L\le o$,
\begin{align*}
\#\{\ell\in(i,j)_N:x_\ell<o\}
&\le\#(\omega_i,o]_O+\#(o,\omega_j]_O-1-\#(o,\omega_j]_O\\
&=\#(\omega_i,o]_O-1\\
&<\#(\omega_i,o]_O.
\end{align*}
This again contradicts Lemma~\ref{lem: cr_i=o iff}.
Hence $\omega_L\le\omega_i\le\omega_j$, so $i\in[L,j]_N$.

Therefore $N_x(o)=[L,j]_N$.

\medskip
\noindent\textsc{Case 3:} $\omega_j=o$.
Since $j\in N_x(o)$, we have $\{j\}\subseteq N_x(o)$.
To show that $N_x(o)\subseteq\{j\}$, let $i\ne j$.
If $\omega_i<\omega_j$, $\omega\in\mathcal{X}$ and
$(\omega_i,o]_O=(\omega_i,o)_O\cup\{o\}$ give
\begin{align*}
\#\{\ell\in(i,j)_N:x_\ell<o\}
&\le\#(i,j)_N\\
&=\#(\omega_i,o)_O\\
&=\#(\omega_i,o]_O-1<\#(\omega_i,o]_O.
\end{align*}
If $\omega_j<\omega_i$, then $\#(\omega_j,o]_O=0$.
Thus Lemmas~\ref{lem: CR=h} and \ref{lem: cr_i=o iff} exclude
$i\ne j$, giving $N_x(o)=\{j\}$.\qed
    \end{prf}

Fix a true preference profile $R^0\in\mathcal{R}^n$ for the rest of the appendix.
The operations $\vee$ and $\wedge$ below are defined using $R^0$.

    \begin{lem}\label{lem: 7, x or y}
        If $x,y\in\mathcal{X}\bigl(\NE(\Gamma^{\CR},R^0)\bigr)$, $i\ne j$, and $y_i=x_j=o$, then
$o\in O_i(x)$ or $o\in O_j(y)$.
    \end{lem}

    \begin{prf}
        Suppose instead that $o\notin O_i(x)$ and $o\notin O_j(y)$, that is $i\notin N_x(o)$ and $j\notin N_y(o)$. 
Without loss of generality, after interchanging $(x,j)$ and $(y,i)$ if necessary, assume $\omega_i<\omega_j$.
If $o\le\omega_i$, by Lemmas~\ref{lem: U,L} and \ref{lem: N_x(o)=}, $N_x(o)=[L,j]_N$ and $\omega_L\le o\le\omega_i<\omega_j$,
contradicting $i\notin N_x(o)$.
If $\omega_j\le o$, by Lemmas~\ref{lem: U,L} and \ref{lem: N_x(o)=}, $N_y(o)=[i,U]_N$ and $\omega_i<\omega_j\le o\le\omega_U$,
contradicting $j\notin N_y(o)$.
Hence $\omega_i<o<\omega_j$.
By Lemmas~\ref{lem: U,L} and \ref{lem: N_x(o)=},
$i\notin N_x(o)$, and $j\notin N_y(o)$, we have 
$N_x(o)=[L,j]_N$, $N_y(o)=[i,U]_N$, and 
$\omega_i<\omega_L\le o\le\omega_U<\omega_j$.
Define
\begin{equation*}
S_x=\{\ell\in N:\omega_L\le\omega_\ell<\omega_j,\ o<x_\ell\},\qquad
S_y=\{\ell\in N:\omega_i<\omega_\ell\le\omega_U,\ y_\ell<o\}.
\end{equation*}
Since  $\omega_i<\omega_L\le o\le\omega_U<\omega_j$, we have  $S_x\cup S_y\subseteq(i,j)_N$.
By the definitions of $L$ and $U$ in Lemma~\ref{lem: U,L},  $\#S_x=\#(o,\omega_j]_O$ and $\#S_y=\#(\omega_i,o]_O$.
Since $\omega_i<o<\omega_j$ and $\omega\in\mathcal{X}$,
\begin{equation*}
\#S_x+\#S_y=\#(\omega_i,\omega_j]_O=\#(i,j)_N+1.
\end{equation*}
Thus, by $S_x\cup S_y\subseteq(i,j)_N$,
\begin{equation*}
\#(S_x\cap S_y)=\#S_x+\#S_y-\#(S_x\cup S_y)
\ge\#(i,j)_N+1-\#(i,j)_N=1.
\end{equation*}
Choose $\ell\in S_x\cap S_y$. By the definitions of $S_x$ and $S_y$, $y_\ell<o<x_\ell$ and $\ell\in N_x(o)\cap N_y(o)$.
By the definition of $N_x(o)$ and $N_y(o)$, $o\in O_\ell(x)\cap O_\ell(y)$.
By Lemma~\ref{lem: NE} $x_\ell \mathbin{R_\ell^0} o$ and $y_\ell \mathbin{R_\ell^0} o$.
Since $y_\ell<o<x_\ell$, we have $x_\ell \mathbin{P_\ell^0} o$ and $y_\ell \mathbin{P_\ell^0} o$,
contrary to Lemma~\ref{lem: single-peaked}.\qed
\end{prf}

For $x,y\in\mathcal{X}$, define the permutation $\pi:N\to N$ by
$\pi(i)=x_{y_i}$ for each $i\in N$.
A \emph{cycle} $C$ of $\pi$ is a set of agents that can be ordered so that
$\pi$ maps each agent to the next and the last agent to the first.

\begin{lem}\label{lem: cycle}
    Let $x,y\in \mathcal{X}\bigl(\NE(\Gamma^{\CR},R^0)\bigr)$ and let $C$ be a cycle of $\pi$ with $\#C\ge2$.
Then either $x_i\mathbin{P_i^0} y_i$ for each $i\in C$ or $y_i\mathbin{P_i^0} x_i$ for each $i\in C$.
\end{lem}

\begin{prf}
    For $i\in C$, $\pi(i)\ne i$, so $y_i=x_{\pi(i)}\ne x_i$.
Suppose $y_i\mathbin{P_i^0} x_i$ for one $i\in C$.
Lemma~\ref{lem: NE} implies $y_i\notin O_i(x)$.
Lemma~\ref{lem: 7, x or y}, applied to $y_i=x_{\pi(i)}$, then gives
$x_{\pi(i)}\in O_{\pi(i)}(y)$.
By Lemma~\ref{lem: NE}, $y_{\pi(i)}\mathbin{R_{\pi(i)}^0}x_{\pi(i)}$.
Since $y_{\pi(i)}\ne x_{\pi(i)}$, $y_{\pi(i)}\mathbin{P_{\pi(i)}^0}x_{\pi(i)}$.
Repeating this argument around $C$ gives
$y_\ell \mathbin{P_\ell^0} x_\ell$ for each $\ell\in C$.
Otherwise, by $x_\ell\ne y_\ell$, $x_\ell \mathbin{P_\ell^0} y_\ell$ for each $\ell\in C$.\qed
\end{prf}

\begin{lem}\label{lem: 9, N_y,N_x}
    Let $x\in\mathcal{X}\bigl(\NE(\Gamma^{\CR},R^0)\bigr)$ and $y\in\mathcal{X}$ satisfy
$y_i\mathbin{R_i^0}x_i$ for each $i\in N$.
If $x_j=y_j=o$ for some $j\in N$ and $o\in O$, then $N_y(o)\subseteq N_x(o)$.
\end{lem}

\begin{prf}

Let $\ell\in N_x(o)\setminus\{j\}$. Then $o\in O_\ell(x)$, so
Lemma~\ref{lem: NE} gives $x_\ell \mathbin{R_\ell^0} o$.
Since $x\in\mathcal{X}$ and $\ell\ne j$, we have $x_\ell\ne x_j=o$.
By $y_\ell R_\ell^0 x_\ell$ and transitivity, $x_\ell \mathbin{P_\ell^0} o$ and $y_\ell \mathbin{P_\ell^0} o$.
Lemma~\ref{lem: single-peaked} rules out
$x_\ell<o<y_\ell$ and $y_\ell<o<x_\ell$.
Since $x_\ell\ne o$ and $y_\ell\ne o$, for each
$\ell\in N_x(o)\setminus\{j\}$,
$y_\ell<o$ if $x_\ell<o$, and 
$o<y_\ell$ if $o<x_\ell$.

\medskip
\noindent\textsc{Case 1:} $\omega_j<o$.
By Lemma~\ref{lem: N_x(o)=}, write $N_x(o)=[j,U]_N$.
For each $\ell\in[j,U]_N\setminus\{j\}$ with $x_\ell<o$, we have $y_\ell<o$.
Thus $\{\ell\in N:\omega_j<\omega_\ell\le\omega_U,\ x_\ell<o\} \subseteq\{\ell\in N:\omega_j<\omega_\ell\le\omega_U,\ y_\ell<o\}$.
By the definition of $U$ in Lemma~\ref{lem: U,L},
\begin{align*}
\#\{\ell\in N:\omega_j<\omega_\ell\le\omega_U,\ y_\ell<o\}
&\ge\#\{\ell\in N:\omega_j<\omega_\ell\le\omega_U,\ x_\ell<o\}\\
&=\#(\omega_j,o]_O.
\end{align*}
If $i\in N_y(o)$, Lemma~\ref{lem: N_x(o)=} applied to $y$ gives
$\omega_j\le\omega_i$.
If $\omega_U<\omega_i$, this inequality gives
\begin{equation*}
\#\{\ell\in(j,i)_N:y_\ell<o\}\ge\#(\omega_j,o]_O.
\end{equation*}
However, by $i\in N_y(o)$ and Lemmas~\ref{lem: CR=h} and
\ref{lem: cr_i=o iff}, $\#\{\ell\in(j,i)_N:y_\ell<o\}<\#(\omega_j,o]_O$,
a contradiction.
Thus $\omega_j\le\omega_i\le\omega_U$, and $i\in N_x(o)$.

\medskip
\noindent\textsc{Case 2:} $o<\omega_j$.
Write $N_x(o)=[L,j]_N$ by Lemma~\ref{lem: N_x(o)=}.
For each $\ell\in[L,j]_N\setminus\{j\}$ with $o<x_\ell$,
we also have $o<y_\ell$. By the definition of $L$,
\begin{align*}
\#\{\ell\in N:\omega_L\le\omega_\ell<\omega_j,\ o<y_\ell\}
&\ge\#\{\ell\in N:\omega_L\le\omega_\ell<\omega_j,\ o<x_\ell\}\\
&=\#(o,\omega_j]_O.
\end{align*}
If $i\in N_y(o)$, Lemma~\ref{lem: N_x(o)=} gives $\omega_i\le\omega_j$.
If $\omega_i<\omega_L$, this inequality gives
\begin{equation*}
\#\{\ell\in(i,j)_N:o<y_\ell\}\ge\#(o,\omega_j]_O.
\end{equation*}
However $\omega_i<\omega_j$ and $i\in N_y(o)$.
By \textsc{Case 2} of the proof of Lemma~\ref{lem: N_x(o)=}, applied to $y$,
\begin{equation*}
\#\{\ell\in(i,j)_N:o<y_\ell\}<\#(o,\omega_j]_O.
\end{equation*}
This is a contradiction. Thus $\omega_L\le\omega_i\le\omega_j$,
and $i\in N_x(o)$.

\medskip
\noindent\textsc{Case 3:} $\omega_j=o$.
Since $x_j=y_j=o$, by Lemma~\ref{lem: N_x(o)=},
$N_y(o)=\{j\}=N_x(o)$.\qed
\end{prf}

\begin{lem}\label{lem: 10 N_x,N_y}
    Let $x,y\in \mathcal{X}\bigl(\NE(\Gamma^{\CR},R^0)\bigr)$, $o\in O$, and $j=x_o$.
If $j\in N_y(o)$, then $N_x(o)\subseteq N_y(o)$.
\end{lem}

\begin{prf}
    Let $k=y_o$.
    For the allocation $x$ and object $o$, write $U_x$ for the agent $U$
in Lemma~\ref{lem: U,L} if $\omega_j<o$, and $L_x$ for the agent $L$
if $o<\omega_j$.
For the allocation $y$ and the same object $o$, write $U_y$ for $U$
if $\omega_k<o$, and $L_y$ for $L$ if $o<\omega_k$.
For $\ell\in N_x(o)\cap N_y(o)\setminus\{j,k\}$, by 
Lemma~\ref{lem: NE}, $x_\ell \mathbin{P_\ell^0} o$ and $y_\ell \mathbin{P_\ell^0} o$.
By Lemma~\ref{lem: single-peaked}, either $x_\ell<o$ and $y_\ell<o$, or $o<x_\ell$ and $o<y_\ell$.

If $j\ne k$, then $j\in N_y(o)$ and $x_j=o\ne y_j$ give
$y_jP_jx_j$. 
Since $y_k=x_j$, agents $j,k$ belong to the same cycle
of the permutation in Lemma~\ref{lem: cycle}. 
Hence $y_k\mathbin{P_k^0}x_k$ and $y_k=o$.
Since $x\in \mathcal{X}\bigl(\NE(\Gamma^{\CR},R^0)\bigr)$, by Lemma~\ref{lem: NE},  $k\notin N_x(o)$.

\medskip
\noindent\textsc{Case 1:} $\omega_j<o$.
By Lemmas~\ref{lem: U,L} and \ref{lem: N_x(o)=}, 
$N_x(o)=[j,U_x]_N$ and $o\le\omega_{U_x}$.

If $\omega_k<o$, write $N_y(o)=[k,U_y]_N$.
Since $j\in N_y(o)$, we have $\omega_k\le\omega_j$.
Suppose $\omega_{U_y}<\omega_{U_x}$.
Each $\ell\in(j,U_y]_N$ belongs to $N_x(o)\cap N_y(o)\setminus\{j,k\}$,
so
\begin{align*}
\#\{\ell\in(j,U_y]_N:x_\ell<o\}
&=\#\{\ell\in(j,U_y]_N:y_\ell<o\}\\
&=\#\{\ell\in(k,U_y]_N:y_\ell<o\}
  -\#\{\ell\in(k,j]_N:y_\ell<o\}\\
&\ge\#(\omega_k,o]_O-\#(k,j]_N\\
&=\#(\omega_k,o]_O-\#(\omega_k,\omega_j]_O\\
&=\#(\omega_j,o]_O.
\end{align*}
However, by the definition of $U_x$ and $\omega_{U_y}<\omega_{U_x}$, $\#\{\ell\in(j,U_y]_N:x_\ell<o\}<\#(\omega_j,o]_O$,
a contradiction. 
Thus $\omega_{U_x}\le\omega_{U_y}$, and
$N_x(o)=[j,U_x]_N\subseteq[k,U_y]_N=N_y(o)$.

If $o<\omega_k$, write $N_y(o)=[L_y,k]_N$.
Since $j\in N_y(o)$, $k\notin N_x(o)$, and $o\le\omega_{U_x}$, we have 
$\omega_{L_y}\le\omega_j$ and $\omega_{U_x}<\omega_k$.
Hence $N_x(o)=[j,U_x]_N\subseteq[L_y,k]_N=N_y(o)$.
If $\omega_k=o$, then $N_y(o)=\{k\}$ and $j\in N_y(o)$ give $j=k$,
contrary to $\omega_j<o$.

\medskip
\noindent\textsc{Case 2:} $o<\omega_j$.
By Lemmas~\ref{lem: U,L} and \ref{lem: N_x(o)=}
$N_x(o)=[L_x,j]_N$ and $\omega_{L_x}\le o$.
If $o<\omega_k$, write $N_y(o)=[L_y,k]_N$.
Since $j\in N_y(o)$, we have $\omega_j\le\omega_k$.
Suppose $\omega_{L_x}<\omega_{L_y}$.
Each $\ell\in[L_y,j)_N$ belongs to $N_x(o)\cap N_y(o)\setminus\{j,k\}$,
so
\begin{align*}
\#\{\ell\in[L_y,j)_N:o<x_\ell\}
&=\#\{\ell\in[L_y,j)_N:o<y_\ell\}\\
&=\#\{\ell\in[L_y,k)_N:o<y_\ell\}
  -\#\{\ell\in[j,k)_N:o<y_\ell\}\\
&\ge\#(o,\omega_k]_O-\#[j,k)_N\\
&=\#(o,\omega_k]_O-\#(\omega_j,\omega_k]_O\\
&=\#(o,\omega_j]_O.
\end{align*}
However, by the definition of $L_x$ and $\omega_{L_x}<\omega_{L_y}$,  $\#\{\ell\in[L_y,j)_N:o<x_\ell\}<\#(o,\omega_j]_O$,
a contradiction. 
Thus $\omega_{L_y}\le\omega_{L_x}$, and $N_x(o)=[L_x,j]_N\subseteq[L_y,k]_N=N_y(o)$.

If $\omega_k<o$, write $N_y(o)=[k,U_y]_N$.
Since $j\in N_y(o)$, $k\notin N_x(o)$, and $\omega_{L_x}\le o$, $\omega_j\le\omega_{U_y}$ and  $\omega_k<\omega_{L_x}$.
Hence $N_x(o)=[L_x,j]_N\subseteq[k,U_y]_N=N_y(o)$.

If $\omega_k=o$, then $N_y(o)=\{k\}$ and $j\in N_y(o)$ give $j=k$,
contrary to $o<\omega_j$.

\medskip
\noindent\textsc{Case 3:} $\omega_j=o$.
By Lemma~\ref{lem: N_x(o)=}, $N_x(o)=\{j\}\subseteq N_y(o)$.\qed
\end{prf}

\begin{prop}\label{prop: join}
    For each $x,y\in \mathcal{X}\bigl(\NE(\Gamma^{\CR},R^0)\bigr)$, $x\vee y\in\mathcal{X}\bigl(\NE(\Gamma^{\CR},R^0)\bigr)$.
\end{prop}
\begin{prf}
Let $z=x\vee y$. We first show that $z\in\mathcal{X}$. 
Define $\pi:N\to N$ by $\pi(i)=x_{y_i}$.
Since $x,y\in\mathcal{X}$, $\pi$ is a permutation, and its cycles partition $N$.
For each cycle $C$, $\pi$ permutes $C$, so
\begin{equation*}
\{y_i\in O:i\in C\}=\{x_{\pi(i)}\in O:i\in C\}=\{x_i\in O:i\in C\}.
\end{equation*}
If $\#C=1$, $y_i=x_{\pi(i)}=x_i$, and hence $z_i=x_i=y_i$.
If $\#C\ge2$, Lemma~\ref{lem: cycle} and the definition of $z$ imply
that either $z_i=x_i$ for each $i\in C$ or $z_i=y_i$ for each $i\in C$.
Thus $\{z_i\in O:i\in C\}=\{x_i\in O:i\in C\}$.
Since $x_i\ne x_j$ for $i\ne j$ and the cycles partition $N$,
$z_i\ne z_j$ for $i\ne j$. Hence $z\in\mathcal{X}$.

Next we show that $z\in\mathcal{X}\bigl(\NE(\Gamma^{\CR},R^0)\bigr)$.
Fix $i\in N$ and $o\in O_i(z)$, and let $j=z_{o}$. 
Since $z_j\in\{x_j,y_j\}$, choose $w\in\{x,y\}$ with $w_j=z_j=o$. 
By the definition of $z$, $z_\ell\mathbin{R_\ell^0}w_\ell$ for each $\ell\in N$.
Since $w\in \mathcal{X}\bigl(\NE(\Gamma^{\CR},R^0)\bigr)$, by Lemma~\ref{lem: 9, N_y,N_x}, $i\in N_z(o)\subseteq N_w(o)$.
Thus $o\in O_i(w)$, and by Lemma~\ref{lem: NE}, $w_i\mathbin{R_i^0}o$. 
By transitivity of $R_i^0$, $z_i\mathbin{R_i^0}w_i\mathbin{R_i^0}o$.
By Lemma~\ref{lem: NE}, $z\in\mathcal{X}\bigl(\NE(\Gamma^{\CR},R^0)\bigr)$.\qed
\end{prf}

\begin{prop}\label{prop: meet}
    For each $x,y\in \mathcal{X}\bigl(\NE(\Gamma^{\CR},R^0)\bigr)$, $x\wedge y\in\mathcal{X}\bigl(\NE(\Gamma^{\CR},R^0)\bigr)$.
    Moreover, $N_{x\wedge y}(o)=N_x(o)\cap N_y(o)$ for each $o\in O$, and $O_i(x\wedge y)=O_i(x)\cap O_i(y)$ for each $i\in N$.
\end{prop}
\begin{prf}
    Let $z=x\wedge y$.
    We first show that $z\in\mathcal{X}$. 
Define $\pi:N\to N$ by $\pi(i)=x_{y_i}$.
For a cycle with one agent, $x_i=y_i=z_i$.
For a cycle $C$ with $\#C\ge2$, by Lemma~\ref{lem: cycle} and the definition
of $z$, either $z_i=x_i$ for each $i\in C$
or $z_i=y_i$ for each $i\in C$.
Since $\{x_i\in O:i\in C\}=\{y_i\in O:i\in C\}$ and the cycles partition $N$,
$z_i\ne z_j$ for $i\ne j$. Hence $z\in\mathcal{X}$.

Next we show that $O_i(x\wedge y)=O_i(x)\cap O_i(y)$ for each $i\in N$. 
Fix $o\in O$ and let $j=z_{o}$.
Without loss of generality, assume $z_j=x_j=o$.
Let $k=y_{o}$.
For the allocation $x$ and object $o$, let $U_x$ and $L_x$ denote the agents
defined in Lemma~\ref{lem: U,L} in their respective cases.
If $j=k$, then $j\in N_y(o)$.
If $j\ne k$, by the definition of $z$, $y_j\mathbin{P_j^0}x_j$.
Since $y_k=x_j$, by Lemma~\ref{lem: cycle}, $y_k\mathbin{P_k^0}x_k$ and $y_k=o$.
By Lemma~\ref{lem: NE}, $k\notin N_x(o)$.
Applying Lemma~\ref{lem: 7, x or y} to $y_k=x_j=o$ gives
$j\in N_y(o)$. 
Thus, in both cases, by Lemma~\ref{lem: 10 N_x,N_y}, $N_x(o)\subseteq N_y(o)$.
Let $\ell\in N_x(o)\setminus\{j\}$.
Then $\ell\in N_y(o)$ and $\ell\ne k$.
By Lemma~\ref{lem: NE}, $x_\ell P_\ell^0 o$ and $y_\ell \mathbin{P_\ell^0} o$.
By Lemma~\ref{lem: single-peaked} and $z_\ell\in\{x_\ell,y_\ell\}$,
either $x_\ell<o$, $y_\ell<o$, and $z_\ell<o$,
or $o<x_\ell$, $o<y_\ell$, and $o<z_\ell$.

If $\omega_j<o$, write $N_x(o)=[j,U_x]_N$.
For each $i\in[j,U_x]_N$,
\begin{equation*}
\{\ell\in(j,i]_N:z_\ell<o\}
=\{\ell\in(j,i]_N:x_\ell<o\}.
\end{equation*}
Hence
\begin{equation*}
\#\{\ell\in(j,i]_N:z_\ell<o\}<\#(\omega_j,o]_O
\end{equation*}
for each $i\in[j,U_x)_N$, and
\begin{equation*}
\#\{\ell\in(j,U_x]_N:z_\ell<o\}
=\#(\omega_j,o]_O.
\end{equation*}
Since $z_j=x_j=o$, by Lemma~\ref{lem: N_x(o)=}, $N_z(o)=[j,U_x]_N=N_x(o)$.

If $o<\omega_j$, write $N_x(o)=[L_x,j]_N$.
For each $i\in[L_x,j]_N$,
\begin{equation*}
\{\ell\in[i,j)_N:o<z_\ell\}
=\{\ell\in[i,j)_N:o<x_\ell\}.
\end{equation*}
Hence
\begin{equation*}
\#\{\ell\in[i,j)_N:o<z_\ell\}<\#(o,\omega_j]_O
\end{equation*}
for each $i\in(L_x,j]_N$, and
\begin{equation*}
\#\{\ell\in[L_x,j)_N:o<z_\ell\}
=\#(o,\omega_j]_O.
\end{equation*}
Since $z_j=x_j=o$, by Lemma~\ref{lem: N_x(o)=},  $N_z(o)=[L_x,j]_N=N_x(o)$.

If $\omega_j=o$, by Lemma~\ref{lem: N_x(o)=}, $N_z(o)=\{j\}=N_x(o)$.
Together with $N_x(o)\subseteq N_y(o)$, these equalities give $N_z(o)=N_x(o)\cap N_y(o)$ 
for each $o\in O$.
By the definitions of these sets,
$O_i(z)=O_i(x)\cap O_i(y)$ for each $i\in N$.

Finally, we show that $z\in\mathcal{X}\bigl(\NE(\Gamma^{\CR},R^0)\bigr)$. 
Fix $i\in N$ and $o\in O_i(z)$.
Then $o\in O_i(x)\cap O_i(y)$.
Since $x,y\in \mathcal{X}\bigl(\NE(\Gamma^{\CR},R^0)\bigr)$, by Lemma~\ref{lem: NE},  $x_i\mathbin{R_i^0} o$ and $y_i\mathbin{R_i^0} o$.
Since $z_i\in\{x_i,y_i\}$, we have $z_i\mathbin{R_i^0} o$.
By Lemma~\ref{lem: NE},  $z\in \mathcal{X}\bigl(\NE(\Gamma^{\CR},R^0)\bigr)$.\qed
\end{prf}

\begin{prfofthm1}
\mbox{}\\
\noindent\textbf{Lattice structure.}
By \emph{strategy-proofness}, $R^0\in\NE(\Gamma^{\CR},R^0)$, so the set of equilibrium allocations is nonempty.
For each $x,y\in\mathcal{X}\bigl(\NE(\Gamma^{\CR},R^0)\bigr)$, by Propositions \ref{prop: join} and \ref{prop: meet}, $x\vee y,x\wedge y\in\mathcal{X}\bigl(\NE(\Gamma^{\CR},R^0)\bigr)$.
Thus $\bigl(\mathcal{X}\bigl(\NE(\Gamma^{\CR},R^0)\bigr),\succeq_{R^0}\!\bigr)$ is a lattice.

\medskip
\noindent\textbf{Least element.}
For each $i\in N$ and each $R'_{-i}\in\mathcal{R}^{n-1}$,
\emph{individual rationality} gives
$\CR_i(R_i^{\omega_i},R'_{-i})\mathbin{R_i^{\omega_i}}\omega_i$.
Since $\omega_i$ is the most preferred object under $R_i^{\omega_i}$,
$\CR_i(R_i^{\omega_i},R'_{-i})=\omega_i$.
In particular, $\CR(R^\omega)=\omega$.
If only $i$ changes their report to $R'_i\in\mathcal{R}$,
the same argument gives
$\CR_\ell(R'_i,R^\omega_{-i})=\omega_\ell$ for each $\ell\ne i$.
Since
$O\setminus\{\omega_\ell:\ell\in N\setminus\{i\}\}=\{\omega_i\}$, $\CR_i(R'_i,R^\omega_{-i})=\omega_i$.
Thus $R^\omega\in\NE(\Gamma^{\CR},R^0)$ and
$\omega\in\mathcal{X}\bigl(\NE(\Gamma^{\CR},R^0)\bigr)$.
For each $x\in\mathcal{X}\bigl(\NE(\Gamma^{\CR},R^0)\bigr)$ and each $i\in N$,
$\CR_i(R_i^{\omega_i},R^x_{-i})=\omega_i$, so $\omega_i\in O_i(x)$.
By Lemma~\ref{lem: NE}, $x_i\mathbin{R_i^0}\omega_i$.
Hence $\omega$ is the least element.

\medskip
\noindent\textbf{Greatest element.}
By \emph{strategy-proofness}, for each $i\in N$ and each $R'_i\in\mathcal{R}$,
$\CR_i(R^0)\mathbin{R_i^0}\CR_i(R'_i,R_{-i}^0)$. 
Hence $R^0\in\NE(\Gamma^{\CR},R^0)$ and $\CR(R^0)\in\mathcal{X}\bigl(\NE(\Gamma^{\CR},R^0)\bigr)$.
For each $x\in \mathcal{X}\bigl(\NE(\Gamma^{\CR},R^0)\bigr)$, 
the least upper bound satisfies $\CR(R^0)\vee x\in\mathcal{X}\bigl(\NE(\Gamma^{\CR},R^0)\bigr)\subseteq\mathcal{X}$ and $[\CR(R^0)\vee x]_i\mathbin{R_i^0}\CR_i(R^0)$ for  each $i\in N$.
By \emph{efficiency} under $R^0$, $\CR(R^0)\vee x=\CR(R^0)$.
Since this least upper bound also weakly Pareto dominates $x$,
$\CR(R^0)\succeq_{R^0}x$. 
Thus $\CR(R^0)$ is the greatest element. 

\medskip
\noindent\textbf{Completeness.}
Since $0<\#\mathcal{X}\bigl(\NE(\Gamma^{\CR},R^0)\bigr)<\infty$, taking least upper bounds and greatest lower bounds two elements at a time
gives both bounds for each
$\emptyset\ne \mathcal{X}'\subseteq\mathcal{X}\bigl(\NE(\Gamma^{\CR},R^0)\bigr)$.
The empty set has least upper bound $\omega$ and greatest lower bound
$\CR(R^0)$.
Hence $\bigl(\mathcal{X}\bigl(\NE(\Gamma^{\CR},R^0)\bigr),\succeq_{R^0}\bigr)$
is a complete lattice, with
\begin{equation*}
\bigvee\mathcal{X}\bigl(\NE(\Gamma^{\CR},R^0)\bigr)=\CR(R^0),\qquad
\bigwedge\mathcal{X}\bigl(\NE(\Gamma^{\CR},R^0)\bigr)=\omega.
\end{equation*}\qed
\end{prfofthm1}

\section*{Acknowledgments}
I am grateful to Ryosuke Sakai and Takehiko Yamato for their helpful comments and suggestions.
\bibliographystyle{apalike}
\bibliography{crawler_NE_LATTICE}

@article{bade2019matching,
  title={Matching with single-peaked preferences},
  author={Bade, Sophie},
  journal={Journal of Economic Theory},
  volume={180},
  pages={81--99},
  year={2019},
  publisher={Elsevier}
}

@article{tamura2022object,
  title={Object reallocation problems under single-peaked preferences: Two characterizations of the crawler},
  author={Tamura, Yuki},
  journal={International Journal of Game Theory},
  volume={51},
  number={3},
  pages={537--565},
  year={2022},
  publisher={Springer}
}

@article{bochet2020dominance,
  title     = {Dominance of truthtelling and the lattice structure
               of {Nash} equilibria},
  author    = {Bochet, Olivier and Tumennasan, Norovsambuu},
  journal   = {Journal of Economic Theory},
  volume    = {185},
  pages     = {104952},
  year      = {2020},
  publisher = {Elsevier},
  note      = {Includes an online appendix}
}

@article{arribillaga2023all,
  title={All sequential allotment rules are obviously strategy-proof},
  author={Arribillaga, R Pablo and Mass{\'o}, Jordi and Neme, Alejandro},
  journal={Theoretical Economics},
  volume={18},
  number={3},
  pages={1023--1061},
  year={2023},
  publisher={Wiley Online Library}
}

@article{shapley1974cores,
  title={On cores and indivisibility},
  author={Shapley, Lloyd and Scarf, Herbert},
  journal={Journal of mathematical economics},
  volume={1},
  number={1},
  pages={23--37},
  year={1974},
  publisher={Elsevier}
}

@article{roth1982incentive,
  title={Incentive compatibility in a market with indivisible goods},
  author={Roth, Alvin E},
  journal={Economics letters},
  volume={9},
  number={2},
  pages={127--132},
  year={1982},
  publisher={Elsevier}
}

@article{tamura2022crawler,
  title={The crawler: Three equivalence results for object (re) allocation problems when preferences are single-peaked},
  author={Tamura, Yuki and Hosseini, Hadi},
  journal={Journal of Economic Theory},
  volume={203},
  pages={105466},
  year={2022},
  publisher={Elsevier}
}

@article{ma1994strategy,
  title={Strategy-proofness and the strict core in a market with indivisibilities},
  author={Ma, Jinpeng},
  journal={International Journal of Game Theory},
  volume={23},
  number={1},
  pages={75--83},
  year={1994},
  publisher={Springer}
}

@article{barbera1997strategy,
  title={Strategy-proof allotment rules},
  author={Barber{\`a}, Salvador and Jackson, Matthew O and Neme, Alejandro},
  journal={Games and Economic Behavior},
  volume={18},
  number={1},
  pages={1--21},
  year={1997},
  publisher={Elsevier}
}

@article{li2017obviously,
  title={Obviously strategy-proof mechanisms},
  author={Li, Shengwu},
  journal={American Economic Review},
  volume={107},
  number={11},
  pages={3257--3287},
  year={2017},
  publisher={American Economic Association 2014 Broadway, Suite 305, Nashville, TN 37203}
}

@article{sprumont1991division,
  title={The division problem with single-peaked preferences: A characterization of the uniform allocation rule},
  author={Sprumont, Yves},
  journal={Econometrica},
  volume={59},
  number={2},
  pages={509--519},
  year={1991},
  publisher={JSTOR}
}

\end{document}